\documentclass{article}

\usepackage{amssymb}
\usepackage{graphicx}
\usepackage{cases}

\newtheorem{problem}{Problem}
\newtheorem{lemma}{Lemma}
\newtheorem{proof}{Proof}
\newtheorem{example}{Example}
\newtheorem{definition}{Definition}

\begin{document}
\title{GTRACE-RS: Efficient Graph Sequence Mining \\using Reverse Search}
\author{
Akihiro Inokuchi, Hiroaki Ikuta, and Takashi Washio\\
The Institute of Scientific and Industrial Research, 
Osaka University\\
8-1 Mihogaoka, Ibaraki, Osaka 567-0047 JAPAN\\
inokuchi@ar.sanken.osaka-u.ac.jp\\
}
\date{}
\maketitle

\begin{abstract}
The mining of frequent subgraphs from labeled graph data has been studied extensively. Furthermore, much attention has recently been paid to frequent pattern mining from graph sequences. A method, called GTRACE, has been proposed to mine frequent patterns from graph sequences under the assumption that changes in graphs are gradual. Although GTRACE mines the frequent patterns efficiently, it still needs substantial computation time to mine the patterns from graph sequences containing large graphs and long sequences. In this paper, we propose a new version of GTRACE that enables efficient mining of frequent patterns based on the principle of a reverse search. The underlying concept of the reverse search is a general scheme for designing efficient algorithms for hard enumeration problems. Our performance study shows that the proposed method is efficient and scalable for mining both long and large graph sequence patterns and is several orders of magnitude faster than the original GTRACE.
\end{abstract}

\section{Introduction}
\label{Sec:Intro}
Studies on data mining have established many approaches for finding characteristic patterns from a variety of structured data. 
Graph mining, which efficiently mines all subgraphs appearing more frequently than a given threshold from a set of graphs, focuses on the topological relations between vertices in the graphs~\cite{cookBook2}. 
AGM~\cite{inokuchi-pkdd00}, gSpan~\cite{yan02}, and Gaston~\cite{nijssen04} mine frequent subgraphs, starting with those of size~1, using the anti-monotonicity of the support values. 
Although the major algorithms for graph mining are quite efficient in practice, they require substantial computation time to mine complex frequent subgraphs, owing to the NP-completeness of subgraph isomorphism matching~\cite{garey79}. Accordingly, the conventional methods are not suitable for very complex graphs, such as graph sequences.

Graph sequences, however, are used extensively to model objects in many real-world applications.
For example, a human network can be represented as a graph, where a human and the relationship between two humans correspond to a vertex and an edge, respectively. If a human joins (or leaves) the community in the human network, the numbers of vertices and edges in the graph increase (or decrease). Similarly, a gene network consisting of genes and their interactions produces a graph sequence in the course of its evolutionary history by acquiring new genes, deleting genes, and mutating genes. 

\begin{figure}[b]
\centering
\includegraphics[height=13.3cm]{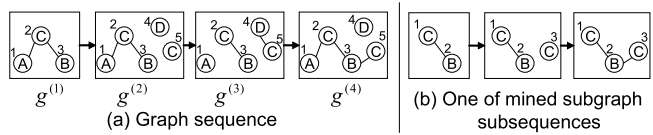}
\vspace{-114mm}
\caption{Examples of a graph sequence and subgraph subsequence for mining.}
\label{Pattern1}
\end{figure}

Recently, much attention has been paid to relevant frequent pattern mining\footnote{The relevancy of frequent patterns is defined in Section~2.2.} from graph sequences~\cite{inokuchi2008} (dynamic graphs~\cite{Borgwardt} or evolving graphs~\cite{Berlingerio}). 
Figure~\ref{Pattern1}(a) shows an example of a graph sequence containing 4 steps and 5 vertex IDs, denoted by the numbers attached to the vertices. 
The problem we address in this paper is how to mine patterns, as shown in Fig.~\ref{Pattern1}(b), that appear more frequently than a given threshold from a set of graph sequences. In~\cite{inokuchi2008}, Inokuchi and Washio proposed transformation rules (TRs) for representing graph sequences compactly under the assumption that the change in each graph in the graph sequence is gradual. In other words, only a small part of the structure changes, while the other part remains unchanged between two successive graphs $g^{(j)}$ and $g^{(j+1)}$ in the graph sequence. 
For example, the change between successive graphs $g^{(j)}$ and $g^{(j+1)}$ in the graph sequence shown in Fig.~\ref{change} is represented as an ordered sequence of two TRs $\langle vi_{[1,A]}^{(j)},ed_{[(2,3),\bullet]}^{(j)} \rangle$. 
This sequence of TRs implies that a vertex with vertex ID 1 and label $A$ is inserted~($vi$), and then the edge between the vertices with vertex IDs 2 and 3 is deleted ($ed$). By assuming that the change in each graph is gradual, we can represent a graph sequence compactly even if the graph in the graph sequence has many vertices and edges. Based on this idea, Inokuchi and Washio proposed a method, called GTRACE (\underline{g}raph \underline{tra}nsformation sequen\underline{ce} mining), for efficiently mining all frequent patterns, called relevant FTSs (frequent transformation subsequences), from ordered sequences of TRs~\cite{inokuchi2008}. 
In a similar manner to PrefixSpan~\cite{yan02}, GTRACE first recursively mines FTSs, appending a TR to the tail of the mined FTS, and then removes irrelevant FTSs during post-processing. Since most of the FTSs mined from graph sequences by GTRACE are irrelevant, if we mine only relevant FTSs from the graph sequences, we can greatly reduce the computation time for this mining process, thus enabling it to be applied to graph sequences containing large graphs and long sequences.  

Our objective graph sequence is more general than both the dynamic graph and evolving graph, and GTRACE and the proposed method in this paper are applicable to both dynamic graphs and evolving graphs, although methods~\cite{Borgwardt,Berlingerio} for mining relevant frequent patterns are not applicable to graph sequences.
In~\cite{Borgwardt}, Borgwardt et al.~proposed a method for mining relevant frequent patterns from a graph sequence represented by a dynamic graph. They assumed that the number of edges in a dynamic graph increases and decreases, while the number of vertices remains constant. They also assumed that labels assigned to vertices in the dynamic graph do not change and that no labels are assigned to edges. On the other hand, Berlingerio et al.~proposed a method to mine relevant frequent patterns from a graph sequence represented by an evolving graph~\cite{Berlingerio}. They assumed that the numbers of vertices and edges in an evolving graph increase, but do not decrease, and that labels assigned to vertices and edges in the dynamic graph do not change. In addition, a vertex in an evolving graph always comes with an edge connected to the vertex.

In this paper, we propose a new version of GTRACE that enables more efficient mining of only relevant FTSs based on the principle of a reverse search~\cite{Avis}. Our performance study shows that the proposed method is efficient and scalable for mining both long and large graph sequence patterns, and is several orders of magnitude faster than the original GTRACE.

\begin{figure}[t]
\centering
\includegraphics[height=9.7cm]{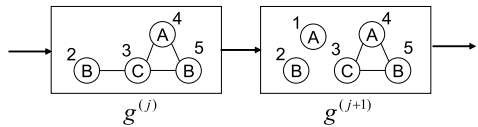}
\vspace{-80mm}
\caption{Change between two successive graphs.}
\label{change}
\end{figure}



\subsection{Frequent Graph Mining}
\label{sec:FrequentSubgraphMining}
Graph mining is the task of finding novel, useful, and ``understandable'' graph-theoretic patterns in a graph representation of data~\cite{cookBook}. Frequent graph mining is a representative task in graph mining that efficiently mines all subgraphs that appear more frequently than a given threshold from a set of labeled graphs.

A graph database $DB$ is a set of tuples $\langle gid, g\rangle$, where $gid$ is a graph ID and $g$ is a labeled graph. A tuple $\langle gid, g\rangle$ is said to contain a graph $p$, if $p$ is a subgraph of $g$, {\it i.e.}, $p \sqsubseteq g$. The support of graph $p$ in database $DB$ is the number of tuples in the database containing $p$, {\it i.e.}, 
$
\sigma(p)=|\{gid \mid (\langle gid, g\rangle \in DB) \wedge (p \sqsubseteq g)\}|.
$
Given a positive integer $\sigma'$ as the support threshold, a graph $p$ is called a ``frequent subgraph'' pattern in the graph database $DB$, if at least $\sigma'$ tuples in the database contain $p$, {\it i.e.}, $\sigma(p)\ge \sigma'$. Representative methods for frequent graph mining, such as AGM~\cite{inokuchi-pkdd00}, gSpan~\cite{yan02}, and Gaston~\cite{nijssen04}, mine frequent subgraphs starting with those of size 1, using the anti-monotonicity of the support values. 

We briefly review gSpan, because it is used to implement the method proposed in this paper. Here, for the sake of simplicity, we assume that edges in graphs have labels, whereas vertices in the graphs do not.
Given a frequent pattern $p$ with $n$ vertices and $k$ edges, vertices in $p$ are traversed in a depth first manner to assign vertex IDs from $v_1$ to $v_n$. The starting vertex and the last visited vertex in the traversal are called the root $v_1$ and the rightmost vertex $v_n$, respectively, while the straight path from $v_1$ to $v_n$ is called the rightmost path. According to the traversal, $p$ is represented by DFS code consisting of triplets $(u,u',l)$, where $l$ is an edge label between vertices $v_u$ and $v_{u'}~(1\le u,u'\le n)$. The linear order of the DFS codes is defined as follows. For DFS codes $\alpha=(a_1, a_2,\cdots,a_k)$ and $\beta=(b_1, b_2,\cdots,b_h)$, $\alpha \preceq \beta$, iff either of the following conditions is true:
\begin{itemize}
\item $\exists t, 1 \le t \le min(k,h), a_q=b_q~for~q<t, a_t \prec_e b_t$ 
\item $a_q=b_q~for~0 \le q \le k,~and~k \le h$
\end{itemize}
where $\prec_e$ is the linear order among the triplets $(u,u',l)$. Since there are many DFS codes for an identical $p$, the minimal DFS code of the DFS codes representing $p$ is defined as the canonical code for $p$.

\begin{figure}[t]
\centering
\includegraphics[height=58mm]{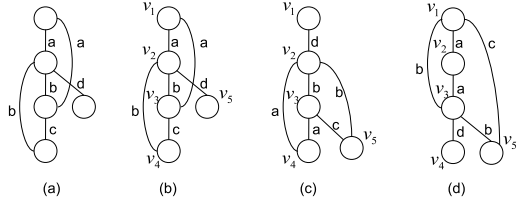}
\vspace{-30mm}
\caption{Depth-first traversal for a frequent pattern.}
\label{dfstree}
\end{figure}

\begin{table}[t]
\caption{DFS Codes for Fig.~\ref{dfstree}.}
\begin{center}
\begin{tabular}{l|l|l|l} \hline
edge $k$ & $\alpha$ (Fig.~\ref{dfstree}(b))  & $\beta$ (Fig.~\ref{dfstree}(c))  & $\gamma$ (Fig.~\ref{dfstree}(d)) \\ \hline \hline
1 & $(1,2,a)$ & $(1,2,d)$ & $(1,2,a)$ \\
2 & $(2,3,b)$ & $(2,3,b)$ & $(2,3,a)$ \\
3 & $(3,1,a)$ & $(3,4,a)$ & $(3,1,b)$ \\
4 & $(3,4,c)$ & $(4,2,a)$ & $(3,4,d)$ \\
5 & $(4,2,b)$ & $(3,5,c)$ & $(3,5,b)$ \\
6 & $(2,5,d)$ & $(5,2,b)$ & $(5,1,c)$ \\ \hline
\end{tabular}
\end{center}
\label{DFSCODE}
\end{table}

\begin{example}
Table~\ref{DFSCODE} gives the DFS codes for the different traversals shown in Fig.~\ref{dfstree}(b)-(d) of the frequent subgraph $p$ depicted in Fig.~\ref{dfstree}(a). According to the linear order of the DFS codes, $\gamma \prec \alpha \prec \beta$, and $\gamma$ is the canonical code for $p$.
\end{example}

A frequent subgraph $p$ is extended based on the pattern-growth principle. Given a DFS code $\alpha=(a_1, a_2,\cdots,a_k)$ for $p$ with $k$ edges, $p$ is extended by adding a new edge to obtain a new pattern with $k+1$ edges, which is represented by $\alpha'=(a_1, a_2,\cdots,a_k,a_{k+1})$. The new edge can be added between the rightmost vertex and other vertices on the rightmost path (backward extension), or it can introduce a new vertex and connect to vertices on the rightmost path (forward extension). Overall, new edges are only added to the vertices along the rightmost path. With this restricted extension, gSpan reduces the generation of the same frequent subgraphs. However, it still guarantees the completeness of enumerating all frequent subgraphs.

\section{Mining Graph Sequences}
\label{sec:Problem Definition}
\subsection{Representation}
In this section, we briefly review the compilation used to represent graph sequences compactly in GTRACE. Figure~\ref{Pattern1}(a) shows an example of a graph sequence. 
Graph $g^{(j)}$ is the $j$-th labeled graph in the sequence. The problem we address in this paper is how to mine patterns that appear more frequently than a given threshold from a set of graph sequences. 
In~\cite{inokuchi2008}, Inokuchi and Washio proposed TRs to represent graph sequences compactly under the assumption that ``the change over successive graphs is gradual''. In other words, only a small part of the graph changes between two successive graphs $g^{(j)}$ and $g^{(j+1)}$ in a graph sequence, while the other parts remain unchanged. 
In the aforementioned human and gene networks, this assumption certainly holds, because most of the changes in vertices are progressive over successive steps. 
A direct representation of a graph sequence is not compact, because many parts of the graph remain unchanged over several steps and are therefore redundant in the representation.
On the other hand, a graph sequence can be compactly represented by introducing a representation of graph transformation based on rules for insertion, deletion, and relabeling of vertices and edges under the gradual change assumption.

A labeled graph $g$ is represented as $g=(V,E,L,f)$, where $V=\{v_1,\cdots,v_z\}$ is a set of vertices, $E=\{(v,v')\mid (v,v') \in  V\times V\}$ is a set of edges, and $L$ is a set of labels such that $f:V \cup E \rightarrow L$. 
A graph sequence is represented as $d=\langle g^{(1)}~g^{(2)}\cdots g^{(n)}\rangle$.
We assume that each vertex $v$ is mutually distinct from the others in $g^{(j)}$ and has a vertex ID $id(v)$ in $d$. We define the set of vertex IDs to be $ID_V(d)=\{id(v) \mid v\in V(g^{(j)}), g^{(j)} \in d\}$ and the set of pairs of vertex IDs to be $ID_E(d)=\{(id(v),id(v'))\mid (v,v')\in E(g^{(j)}), g^{(j)} \in d\}$.
For example, in the human network mentioned in Section~\ref{Sec:Intro}, each person has a vertex ID, and his/her gender is an example of a vertex label.
To represent a graph sequence compactly, we focus on the differences between two successive graphs $g^{(j)}$ and $g^{(j+1)}$ in the sequence.
\begin{definition}
Given a graph sequence 
$d =\langle g^{(1)} \cdots g^{(n)}\rangle$, 
the differences between $g^{(j)}$ and $g^{(j+1)}$ are interpolated by a virtual sequence 
$d^{(j)}=\langle g^{(j,1)} \cdots  g^{(j,m_j)}\rangle $, where $g^{(j,1)}=g^{(j)}$ and 
$g^{(j,m_j)}=g^{(j+1)}$, such that the edit distance~\cite{Sanfeliu} between any two successive graphs is 1, and in which the edit distance between any two intrastates is the minimum. Therefore, $d$ is represented by the interpolations as $d=\langle d^{(1)}\cdots d^{(n-1)}\rangle$.~~~~~~~~~~~~~~~~~~~~~~~~~~~~~~~~~~~~~~~~~~~~~~~~~~~~~~~~~~~~~~~~~~~~~~~~~~~$\blacksquare$ 
\end{definition}
We call $g^{(j)}$ and $g^{(j,k)}$ an interstate and intrastate, respectively. 
The order of interstates represents the order of graphs in a sequence. On the other hand, the order of intrastates is the order of graphs in the artificial interpolation.

The transformation is represented by the following TR.
\begin{definition}
A TR that transforms $g^{(j,k)}$ to $g^{(j,k+1)}$ is expressed as $tr^{(j,k)}_{[o_{jk},l_{jk}]}$, where 
\begin{itemize}
  \item $tr$ is a transformation type, which is either insertion, deletion, or relabeling of a vertex or an edge,
  \item $o_{jk}$ is an element in $ID_V(d) \cup ID_E(d)$ to be transformed, and 
  \item $l_{jk} \in L$ is the label to be assigned to the element by the transformation.~$\blacksquare$ 
 \end{itemize}
\end{definition}

\begin{table}[t]
\begin{center}
\caption{Transformation rules (TRs) representing graph sequence data.}
\begin{tabular}{|l|l|} \hline
Vertex Insertion 		& Insert a vertex with label $l$ and vertex   \\
$vi^{(j,k)}_{[u,l]}$	& ID $u$ into $g^{(j,k)}$ to transform to $g^{(j,k+1)}$. \\
\hline
Vertex Deletion			    & Delete an isolated vertex with vertex   \\
$vd^{(j,k)}_{[u,\bullet]}$ 	& ID $u$ in $g^{(j,k)}$ to  transform to $g^{(j,k+1)}$.\\

\hline
Vertex Relabeling		& Relabel a vertex with vertex  ID $u$ in  \\ 
$vr^{(j,k)}_{[u,l]}$ 	& $g^{(j,k)}$ to be $l$ to transform to $g^{(j,k+1)}$.\\
\hline
Edge Insertion 				    & Insert an edge with label $l$ between 2     \\
$ei^{(j,k)}_{[(u_1,u_2),l]}$ 	& vertices with vertex IDs $u_1$ and $u_2$ into\\
                                &  $g^{(j,k)}$ to  transform to $g^{(j,k+1)}$.\\
\hline
Edge Deletion	 				  & Delete an edge between 2 vertices with  \\
$ed^{(j,k)}_{[(u_1,u_2),\bullet]}$& vertex IDs $u_1$ and $u_2$ in $g^{(j,k)}$ to \\
&transform to $g^{(j,k+1)}$.\\
\hline
Edge Relabeling					& Relabel an edge between 2 vertices   \\
{$er^{(j,k)}_{[(u_1,u_2),l]}$}	& with vertex IDs $u_1$  and $u_2$ in $g^{(j,k)}$ to be  \\
                                &  $l$ to transform to $g^{(j,k+1)}$.\\
\hline
\multicolumn{2}{l}{Since the transformations of vertex deletion $vd$ and edge }\\
\multicolumn{2}{l}{ deletion $ed$ do not assign any labels to the vertex and the}\\
\multicolumn{2}{l}{edge, respectively, they have dummy arguments $l$, }\\
\multicolumn{2}{l}{represented by `$\bullet$'. }
\end{tabular}
\label{tableTR}
\end{center}
\end{table}

For the sake of simplicity, we denote the TR $tr^{(j,k)}_{[o_{jk},l_{jk}]}$ as $tr^{(j,k)}_{[o,l]}$ by omitting the subscripts for $o_{jk}$ and $l_{jk}$, except where this is likely to cause ambiguity. In~\cite{inokuchi2008}, Inokuchi and Washio introduced six TRs as defined in Table~\ref{tableTR}. In summary, we give the following definition of a transformation sequence.

\begin{figure}[b]
\centering
\includegraphics[height=35mm]{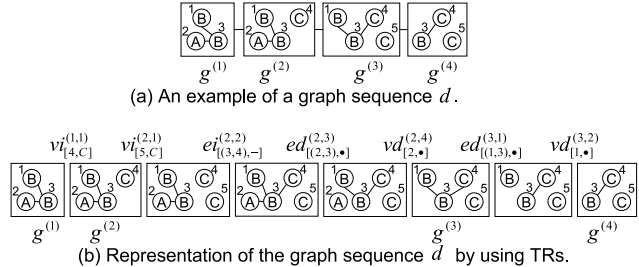}
\caption{A graph sequence with its transformation rules.}
\label{inputdata1}
\end{figure}

\begin{definition}
An intrastate sequence 
$d^{(j)}=\langle 
g^{(j,1)} 
\cdots 
g^{(j,m_j)} 
\rangle$ 
is represented by an ``intrastate transformation sequence''
$
s_d^{(j)}=\langle 
tr^{(j,1)}_{[o,l]} 
\cdots 
tr^{(j,m_j-1)}_{[o,l]} 
\rangle.$ 
Moreover, a graph sequence 
$
d=\langle 
g^{(1)}
\cdots 
g^{(n)}\rangle
$ is represented by an ``interstate transformation sequence'' 
$s_d=\langle 
s_d^{(1)}
\cdots 
s_d^{(n-1)}\rangle$.~~~~~~~~~~~~~~~~~~~~~~~~~~~~~~~~~~~~~~~~~~~~~~~~~~~~~$\blacksquare$ 
\label{def4}
\end{definition}
The notion of an interstate transformation sequence is far more compact than the original graph based representation, because only the differences between two successive interstates appear in the sequence. In addition, in our case, computing a sequence of TRs based on differences between two graphs is solvable in linear time because all vertices have vertex IDs.

\begin{example}
In Fig.~\ref{inputdata1}(a), a graph sequence is expressed as a sequence of insertions and deletions of vertices and edges as shown in Fig.~\ref{inputdata1}(b). 
The sequence is compiled as
$\langle
vi^{(1,1)}_{[4,C]}
vi^{(2,1)}_{[5,C]}
ei^{(2,2)}_{[(3,4),-]} 
ed^{(2,3)}_{[(2,3),\bullet]}
vd^{(2,4)}_{[2,\bullet]} 
ed^{(3,1)}_{[(1,3),\bullet]}
vd^{(3,2)}_{[1,\bullet]}
\rangle$, where ``$-$'' denotes the label of an edge. 
\end{example}

\subsection{Mining Relevant Frequent Transformation Subsequences}
\label{sec:EnumerationProblemandAlgorithm}
In this section, we briefly review how GTRACE mines rFTSs (\underline{r}elevant \underline{f}requent \underline{t}ransformation \underline{s}ubsequences) from a given set of graph sequences. 
To mine rFTSs from a set of transformation sequences, we define an inclusion relation between transformation sequences. 
\begin{definition}
\label{def:occurrence}
Given a transformation sequence $s_p$ of a pattern and a transformation sequence $s_d$ of a data graph sequence $d$,
$s_p$ is a subsequence of $s_d$, denoted as $s_p \sqsubseteq s_d$, iff there is a pair of injective functions $(\phi,\psi)$ satisfying
\begin{itemize}
\item there exist integers $1 \le \phi(1) < \phi(2) <\cdots < \phi(n) \le m$, 
\item there exist integers $\psi(u) \in ID_V(d)$ for vertex IDs $u$ in $s_p$, and
\item $\forall~tr_{[o,l]}^{(j,k)} \in s_p \Rightarrow  \exists k',~ tr_{[o',l]}^{(\phi(j),k')} \in s_d$, where $o'=\psi(u)$, if the TR $tr_{[o,l]}^{(j)}\in s_p$ transforms a vertex with vertex ID $u$. On the other hand, $o'=(\psi(u_1),\psi(u_2))$, if the TR transforms an edge with vertex IDs $u_1$ and $u_2$.~~~~~~~~~~~~~~~~~~~~~~~~~~~~~~~~~~~~~~~~~~~~~~~~~~~~~~~~~~~~~~~~~~~~~~~~~~~~~~~~~~~~~~~~~~~~~~~~~$\blacksquare$
\end{itemize}
\end{definition}
The first condition in Definition~\ref{def:occurrence} states that $\phi$ preserves the order among intrastate transformation sequences, while the second condition states an injective mapping from vertex IDs in $s_p$ to vertex IDs in $s_d$. In addition, the third condition states that a TR corresponding to any TR in $s_p$ must exist in $s_d$. The complexity of finding an occurrence of $s_p$ in $s_d$ is identical to that of subgraph isomorphism matching.

\begin{example} 
Given the graph sequence $d$ in Fig.~\ref{inclusion}(a) represented by the transformation sequence $s_d=
\langle
\underline{vi^{(1,1)}_{[4,C]}} 
           vi^{(2,1)}_{[5,C]} 
\underline{ei^{(2,2)}_{[(3,4),-]} }
\underline{ed^{(2,3)}_{[(2,3),\bullet]}}
\underline{vd^{(2,4)}_{[2,\bullet]}}
           ed^{(3,1)}_{[(1,3),\bullet]}
           vd^{(3,2)}_{[1,\bullet]}
\rangle$, 
the transformation sequence $s_d'
=\langle
vi^{(1,1)}_{[3,C]}
ei^{(2,1)}_{[(2,3),-]} 
ed^{(2,2)}_{[(1,2),\bullet]}
vd^{(2,3)}_{[1,\bullet]}
\rangle$ of the graph sequence $d'$ in Fig.~\ref{inclusion}(b) is a subsequence of $s_d$, and the TRs in $s_d'$ match the underlined rules in $s_d$ via $\phi(j)=j$ for $j \in \{1,2\}$ and $\psi(i)=i+1$ for $i \in ID_V(d')=\{1,2,3\}$.
\label{ex:subsequence}
\end{example}

To mine FTSs consisting of mutually relevant vertices only, the relevancy of vertices and edges is defined as follows\footnote{See~\cite{inokuchi2008,inokuchi2010} for the detail motivations for mining rFTSs.}.
\begin{definition}
Vertex IDs in a graph sequence $d=\langle 
g^{(1)}
\cdots 
g^{(n)}
\rangle$ are relevant to one another, and $d$ is called a ``relevant graph sequence'', if the union graph $g_u(d)$ of $d$ is a connected graph. We define the union graph of $d$ to be $g_u(d)=(V_u,E_u)$, where $V_u=ID_V(d)$ and $E_u=ID_E(d)$.~~~~~~~~~~~~~~~~~~~~~~~~~~~~~~~~~~~~~~~~~~~~~~~~~~~~~$\blacksquare$
\label{uniongraph1}
\end{definition}
Similar to Definition~\ref{uniongraph1}, we define the union graph of the transformation sequence $s_d$. 
\begin{definition}
The union graph $g_u(s_d)=(V_u,E_u)$ of a transformation sequence $s_d$ is similarly defined as
\begin{eqnarray}
V_u&=&\{u \mid tr_{[u,l]}^{(j,k)} \in s_d, tr \in \{vi,vd,vr\}\} \nonumber \\
		&&\cup \{u, u' \mid tr_{[(u,u'),l]}^{(j,k)} \in s_d, tr \in \{ei,ed,er\}\},  \nonumber \\
E_u&=&\{(u, u') \mid tr_{[(u,u'),l]}^{(j,k)} \in s_d, tr \in \{ei,ed,er\}\}.~~~~~~~~~~~~~~~\blacksquare \nonumber
\end{eqnarray}
\label{uniongraph2}
\end{definition}
\begin{example}
Figure~\ref{uniongraphs}(b) shows the union graph of the graph sequence depicted in Fig.~\ref{uniongraphs}(a). In addition, the union graph of a transformation sequence $\langle ei^{(1,1)}_{[(1,2),-]}\\ei^{(2,1)}_{[(2,3),-]}\rangle$ is identical to the graph shown in Fig.~\ref{uniongraphs}(b).
\end{example}

\begin{figure}[t]
\centering
\includegraphics[height=50mm]{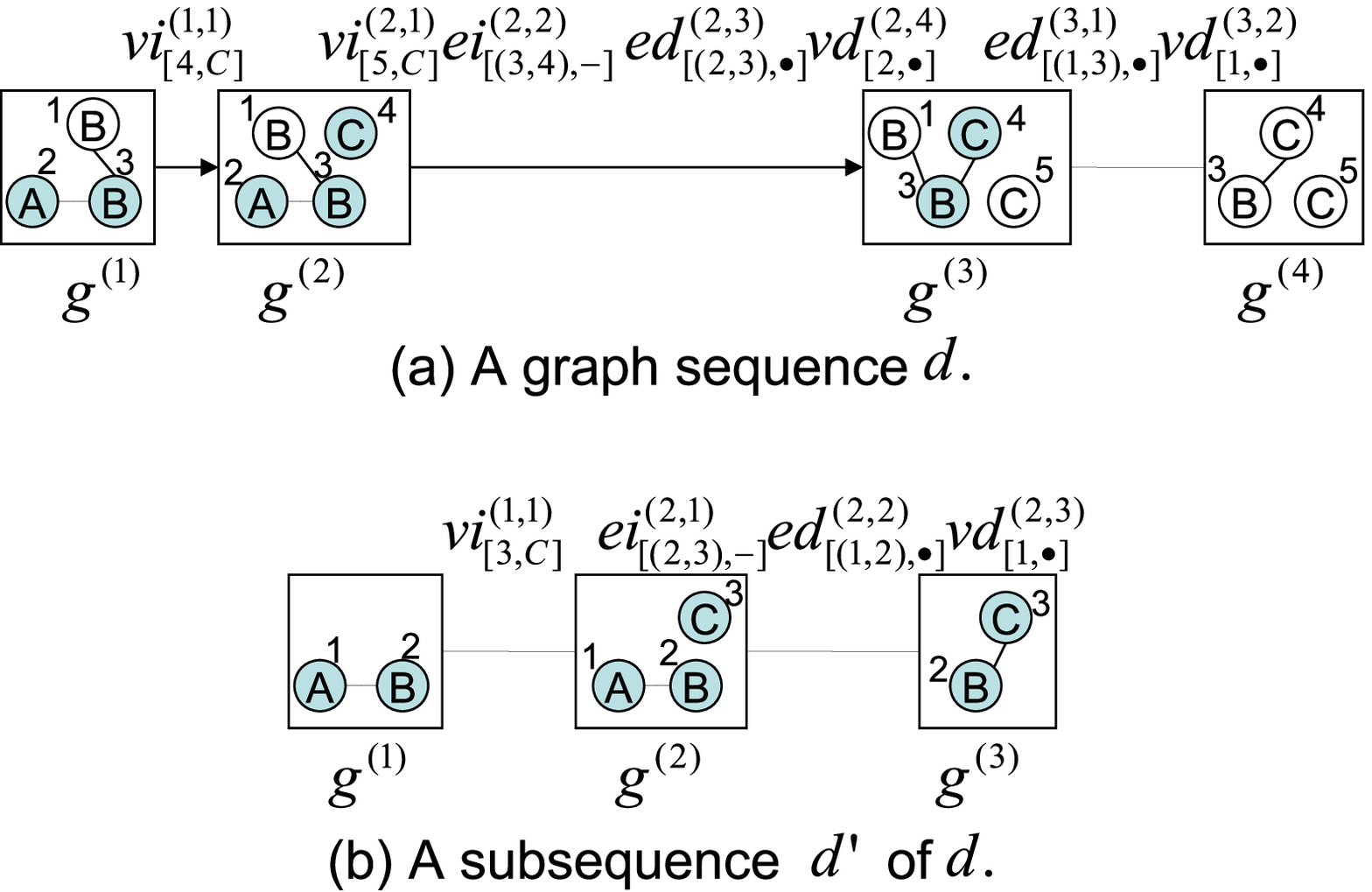}
\caption{Inclusion relation.}
\label{inclusion}
\end{figure}

\begin{figure}[t]
\centering
\includegraphics[height=18mm]{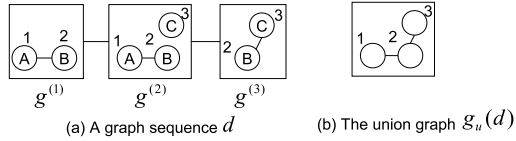}
\caption{Union graph.}
\label{uniongraphs}
\end{figure}

Given a set of data $DB=\{\langle gid, d\rangle \mid d=\langle g^{(1)}~\cdots~g^{(n)}\rangle \}$, the support value $\sigma(s_p)$ of a transformation subsequence $s_p$ is given as 
\[
\sigma(s_p)=|\{gid \mid \langle gid, d\rangle \in DB, s_p \sqsubseteq s_d\}|,
\]
where $s_d$ is the transformation sequence of $d$.
We call a transformation subsequence with support value greater than or equal to a minimum support threshold $\sigma'$ a ``frequent transformation subsequence'' (FTS). The anti-monotonicity of this support value holds; that is, if $s_1 \sqsubset s_2$ then $\sigma(s_1)\ge \sigma(s_2)$.
Using these definitions, we state our mining problem as follows.
\begin{problem}
Given a dataset $DB=\{\langle gid, d\rangle \mid d=\langle g^{(1)} \cdots  g^{(n)}\rangle \}$ and a minimum support threshold $\sigma'$ as input, enumerate all rFTSs.
\end{problem}
\begin{example}
Figure~\ref{miningproblem} shows a graph sequence database $DB$ that contains two graph sequences and nine rFTSs mined from the compiled graph sequences in $DB$ under $\sigma'=2$. In this example, $\langle vi^{(1,1)}_{[1,A]}vi^{(2,1)}_{[2,B]}\rangle $ and $\langle vi^{(1,1)}_{[1,B]}vi^{(2,1)}_{[2,A]}\rangle$ are common subsequences of the compiled graph sequences not mined from $DB$, because their union graphs are not connected.
\end{example}

\begin{figure}[t]
\centering
\includegraphics[height=60mm]{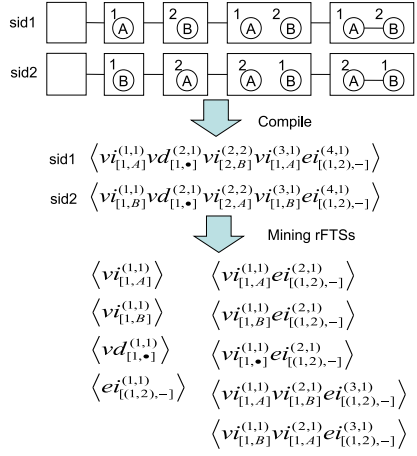}
\caption{Mining rFTSs from $DB$.}
\label{miningproblem}
\end{figure}
To enumerate all rFTSs efficiently, GTRACE recursively mines FTSs by appending a TR to the tail of the current FTS, in a similar manner to PrefixSpan~\cite{pei01}, which is a representative method for mining frequent subsequences from a set of itemset sequences. After mining all the FTSs, GTRACE removes all FTSs whose union graphs are not connected, thus outputting only rFTSs.

\begin{example}
Figure~\ref{GTRACEprocedure} shows the detailed procedure for mining FTSs up to $s_6$ using GTRACE. After mining an FTS $s_i$, GTRACE recursively appends a TR to $s_i$ to mine a longer FTS $s_{i+1}$. After mining all FTSs, 
$s_2$, $s_3$, and $s_4$ are removed during post-processing, because their union graphs are not connected.
\end{example}

\begin{figure}
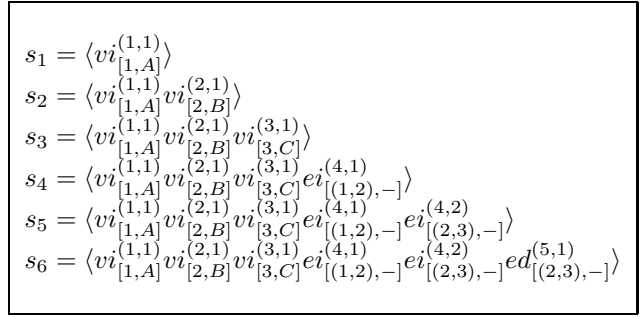

\begin{center}
\begin{tabular}{|l|}\hline 
\\
$s_1=\langle vi_{[1,A]}^{(1,1)}\rangle$ \\
$s_2=\langle vi_{[1,A]}^{(1,1)} vi_{[2,B]}^{(2,1)}\rangle$  \\
$s_3=\langle vi_{[1,A]}^{(1,1)} vi_{[2,B]}^{(2,1)} vi_{[3,C]}^{(3,1)} \rangle$\\
$s_4=\langle vi_{[1,A]}^{(1,1)} vi_{[2,B]}^{(2,1)} vi_{[3,C]}^{(3,1)} 
          ei_{[(1,2),-]}^{(4,1)} \rangle$  \\
$s_5=\langle vi_{[1,A]}^{(1,1)} vi_{[2,B]}^{(2,1)} vi_{[3,C]}^{(3,1)} 
          ei_{[(1,2),-]}^{(4,1)} ei _{[(2,3),-]}^{(4,2)} \rangle $ \\
$s_6=\langle vi_{[1,A]}^{(1,1)} vi_{[2,B]}^{(2,1)} vi_{[3,C]}^{(3,1)} 
          ei_{[(1,2),-]}^{(4,1)} ei _{[(2,3),-]}^{(4,2)} ed_{[(2,3),-]}^{(5,1)}\rangle $ \\
\\
\hline
\end{tabular}
\caption{Mining procedure using GTRACE.}
\label{GTRACEprocedure}
\end{center}
\end{figure}

\subsection{Drawback of GTRACE}
\label{Drawback}
In the previous subsection, we briefly reviewed the problem of graph sequence mining and GTRACE for mining rFTSs from graph sequences. GTRACE first mines a set of FTSs containing a complete set of rFTSs, and then removes any FTS that is not relevant from the set of mined FTSs during post-processing. Since most of the FTSs mined in the first step are not rFTSs, excessive computation time is needed to mine the vast set of FTSs, causing GTRACE to be highly inefficient. 

For example, the FTS $s_6$ shown in Fig.~\ref{GTRACEprocedure} is relevant, because its union graph is connected. To mine this rFTS, GTRACE mines $s_1$, $s_2$, $s_3$, $s_4$, and $s_5$ in order appending a TR to the tail of $s_i$ to obtain a new FTS $s_{i+1}$, where $i=2,\cdots,5$. Then, $s_2$, $s_3$, and $s_4$ are removed during post-processing, because they are irrelevant. Therefore, to mine all rFTSs using GTRACE, a vast set of FTSs, most of which are not rFTSs, first needs to be mined, resulting in inefficient execution of GTRACE. 
In our experiment, 95\% of FTSs mined by GTRACE are irrelevant. By designing a new algorithm that mines only rFTSs, we can reduce the computation time for mining the complete set of rFTSs. In the next section, we propose a new method for mining only rFTSs from a set of graph sequences based on the principle of a reverse search.

\section{Proposed Method}
\label{sec:algorithm3}

\subsection{Reverse Search}

To mine all rFTSs efficiently, we propose a method for mining only rFTSs based on a reverse search~\cite{Avis,Asai}. The underlying concept of the reverse search is a general scheme for designing efficient algorithms for hard enumeration problems. Let $S$ be the set of solutions. In a reverse search, we define the parent-child relation $P \in S \times S$ such that each $x \in S$ has a unique parent $P(x) \in S$. Using the search tree $T$ over $S$ defined by $P$, we can enumerate all solutions of $S$ without duplicates by traversing the search tree $T$ from the root to the leaves. 
Therefore, if we can define the parent-child relation $P$ such that each $x \in S$ has a ``unique'' parent, we can enumerate a complete set of elements in $S$.

For example, given a set of items $I=\{i_1,\cdots,i_{|I|}\}$ and a set of transactions $D=\{t \mid t\subseteq I\}$,  the frequent itemset mining problem is to enumerate all frequent itemsets $S=\{x \mid x\subseteq I, \sigma(x) \ge \sigma',~\sigma(x)=|\{t\mid x \subseteq t, t \in D\}| \}$, where $\sigma'$ is the minimum support threshold. To enumerate the complete set of frequent itemsets efficiently, the parent-child relation $P$ is defined as the unique itemset $x'=\{i_1,i_2, \cdots ,i_{k-1}\}$ derived from $x =\{i_1,i_2, \cdots ,i_{k-1},i_k\}$ by removing the last item $i_k$ in $x$, where the items in $x$ and $x'$ are sorted according to the linear order of the items.
 Starting from the root node corresponding to an empty set, traversing the search tree $T$ over $S$ defined by $P$ enables us to enumerate a complete set of frequent itemsets without duplicates. In the traversal, $P^{-1}(x')$ is used to enumerate itemsets from itemset $x'\in S$ by adding an item to $x'$. Figure~\ref{searchtree} shows the enumeration tree for the itemset $I=\{A,B,C,D\}$, with $2^I$ nodes in the tree. By traversing from the root node in a reverse direction to the arrows in the tree, all itemsets are enumerated. 

\begin{figure}[t]
\centering
\includegraphics[height=35mm]{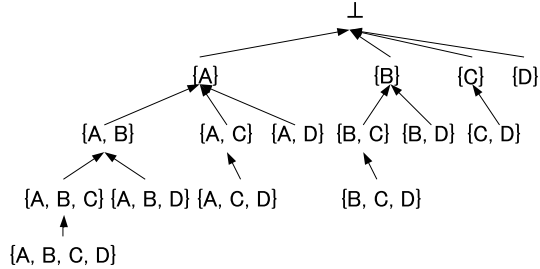}
\caption{Enumeration Tree for items $\{A,B,C,D\}$.}
\label{searchtree}
\end{figure}

However, as mentioned in Section~\ref{Drawback}, if the last TR in an rFTS $s_{i}~(i=2,\cdots,6)$ is removed to derive its parent FTS $s_{i-1}$, $s_{i-1}$ is not always relevant, although $s_{i-1}$ is always frequent according to the anti-monotonicity of the support values. Therefore, the original GTRACE needs all the FTSs to enumerate all the rFTSs, which results in inefficient execution of GTRACE. In this section, to enumerate only rFTSs, we define canonical forms of rFTSs and novel parent-child relations $P_1$, $P_2$, and $P_3$ between the canonical rFTSs. 
 
\begin{definition}
\label{def:canonical}
Given an rFTS $s$, TRs $a_k$, $a_{k-1}, \cdots, a_1$ in $s$ are removed by either $P_1$, $P_2$, or $P_3$ in order as defined later. We define a code $\alpha=(a_1, \cdots, a_k)$ for $s$, where vertex IDs in $g_u(s)$ are assigned in a depth-first manner similarly to gSpan.
For codes $\alpha=(a_1,\cdots,a_k)$ and $\beta=(b_1,\cdots,b_h)$, $\alpha \preceq \beta$, iff either of the following conditions is true:
\begin{itemize}
\item $\exists t, 1 \le t \le min(k,h), a_q=b_q~for~q<t, a_t \prec_{tr} b_t$ 
\item $a_q=b_q~for~0 \le q \le k,~and~k \le h,$
\end{itemize}
where $\prec_{tr}$ is the linear order among TRs $tr_{[o,l]}^{(j,k)}$. Since there are many representations for an identical transformation sequence, the representation corresponding to the minimal code among the representations of an identical transformation sequence is defined as canonical.~~~~~~~~~~~~~~~~~~~~~~~~~~~~~~~~~~~~~~~~~~~~~~~~~~~~~~~~~$\blacksquare$
\end{definition}
In Definition~\ref{def:canonical}, $P_1$, $P_2$, and $P_3$ are parent-child relations 
among canonical rFTSs $S$, which are defined as follows.
\begin{definition}
Given an rFTS $s \in S$ containing TRs applied to vertices, 
we define function $P_1$ mapping from $s$ to $s'$. The transformation sequence $s'$ is derived from $s$ 
by removing the TR located in the last position of all the TRs applied to vertices in $s$.
~~~~~~~~~~~~~~~~~~~~~~~~~~~~~~~~~~~~~~~~~~~~~~~~~~~~~~~~~~~~~~~~$\blacksquare$
\end{definition}
If the length of transformation sequence $s$ is defined as the number of TRs in $s$, the following lemma is obtained.
\begin{lemma}
\label{lemmaP1-1}
Given an rFTS $s \in S$ containing TRs that are applied to vertices and with length greater than 1, $g_u(P_1(s))=g_u(s)$.~~~~~~~~~~~~~~~~~~~~~~~~~~~~~~~~~~~~~~~~$\blacksquare$
\end{lemma}
\begin{proof}
The union graph of $s$ is a connected graph, because $s$ is relevant. If the vertex ID, to which the TR $r$ removed by $P_1$ is applied, is $u$, a TR to transform an edge, whose terminal vertex ID is $u$, must exist in $s$, because $g_u(s)$ is connected. Therefore, vertex~$u$ remains in $g_u(P_1(s))$ after $r$ is removed, and the union graph of transformation sequence $P_1(s)$ 
is isomorphic with $g_u(s)$. On the other hand, given an rFTS $s \in S$ of length 1 containing a TR that is applied to a vertex, $P_1(s)=\bot$.
\end{proof}
According to Lemma~\ref{lemmaP1-1}, $P_1(s)$ is always relevant for an rFTS $s$ containing TRs applied to vertices. In addition, the transformation sequence $P_1(s)$ is frequent according to the anti-monotonicity of the support value and canonical according to Definition~\ref{def:canonical}. Therefore, $P_1(s)$ returns a ``unique'' rFTS in $S$.
\begin{example}
Given $s_6$ in Fig.~\ref{GTRACEprocedure}, $P_1(s_6)$ and $P_1(P_1(s_6))$ are 
\begin{eqnarray}
&&\langle vi_{[1,A]}^{(1,1)} vi_{[2,B]}^{(2,1)} ei_{[(1,2),-]}^{(3,1)} ei _{[(2,3),-]}^{(3,2)} ed_{[(2,3),-]}^{(4,1)}\rangle  \verb| and|  \nonumber \\
&&\langle vi_{[1,A]}^{(1,1)}ei_{[(1,2),-]}^{(2,1)} ei _{[(2,3),-]}^{(2,2)} ed_{[(2,3),-]}^{(3,1)}\rangle , \nonumber
\end{eqnarray}
respectively. Union graphs of $s_6$, $P_1(s_6)$, and $P_1(P_1(s_6))$ are isomorphic to the graph shown in Fig.~\ref{uniongraphs}(b).
\end{example}
Next, we define the second function of our parent-child relation to enumerate rFTSs.
\begin{definition}
\label{def:P2}
Given an rFTS $s' \in S$ that contains only TRs applied to edges, we define function $P_2$ 
mapping from $s'$ to $s''$. The transformation sequence $s''$ is derived from $s'$ 
by removing the TR $tr_{[o,l]}^{(j,k)}$ located in the last position of all the TRs in $s'$
if a TR $tr_{[o',l']}^{(j',k')}$ exists in $s'$ such that $o=o'$ and $j'<j$.~~~~~~~~~~~~~~$\blacksquare$
\end{definition}
If the rFTS $s$ contains TRs that are applied to vertices, we can obtain an rFTS $s'$ that does not contain TRs applied to vertices by applying $P_1$ to $s$ multiple times. In addition, $P_2$ is applicable to $s'$, if the length of $s'$ is greater than the number of edges in $g_u(s')$, 
since at least one TR needs to be applied to each edge in $E(g_u(s'))$.
According to the above definition, the following lemma is obtained.
\begin{lemma}
\label{lemmaP2-1}
Given an rFTS $s' \in S$ containing only TRs applied to edges and whose length is greater than $|E(g_u(s'))|$, $g_u(P_2(s'))$ is identical to $g_u(s')$.~~~~~~~~~~~~~~$\blacksquare$
\end{lemma}
This lemma is proven in the same way as Lemma~\ref{lemmaP1-1}. According to Lemma~\ref{lemmaP2-1}, $P_2(s')$ is always relevant for an rFTS $s'$ containing only TRs applied to edges. In addition, the transformation sequence $P_2(s')$ is frequent according to the anti-monotonicity of the support value and canonical according to Definition~\ref{def:canonical}. Therefore, $P_2(s')$ returns a ``unique'' rFTS in $S$.
\begin{example}
Given $s'_3=\langle ei_{[(1,2),-]}^{(1,1)} ei _{[(2,3),-]}^{(1,2)} ed_{[(2,3),-]}^{(2,1)}\rangle$, $P_2(s'_3)=\langle ei_{[(1,2),-]}^{(1,1)} ei _{[(2,3),-]}^{(1,2)}\rangle$. According to Definition~\ref{def:P2}, we cannot apply $P_2$ to $s'_3$.
\end{example}
Finally, we define the third function of our parent-child relation to enumerate rFTSs.
\begin{definition}
Given an rFTS $s'' \in S$ where each TR is applied to a mutually different edge in $g_u(s'')$, we define function $P_3$ 
mapping from $s''$ to $s'''$. The transformation sequence $s'''$ is derived from $s''$ 
by removing the TR located in the last position in $s''$  
keeping the connectivity of $g_u(s''')$.~~~~~~~~~~~~~~~~~~~~~~~~~~~$\blacksquare$
\label{DeflP3-1}
\end{definition}
If each TR in an rFTS $s$ is not applied to a mutually different edge in $g_u(s)$, we can obtain an rFTS $s''$ where each TR is applied to a mutually different edge in $g_u(s'')$ by applying $P_1$ and $P_2$ to $s$ multiple times.
According to Definition~\ref{DeflP3-1}, $P_3(s'')$ is always relevant for an rFTS $s''$ where each TR is applied to a mutually different edge in $g_u(s'')$, because the connectivity of $g_u(P_3(s''))$ is kept. In addition, the transformation sequence $P_3(s'')$ is frequent according to the anti-monotonicity of the support value and canonical according to Definition~\ref{def:canonical}. Therefore, $P_3(s'')$ returns a ``unique'' rFTS in $S$.

\begin{figure*}
\begin{center}
\begin{tabular}{|ll|}\hline 
&\\
$~~~~~~~P_3(s'_1)=\bot   $ & \\
$s'_1=P_3(s'_2)=\langle~~~~~~~~~~~~~~~~~~~~~~~~ei_{[(1,2),-]}^{(1,1)}  \rangle$ 
& $ \in P_3^{-1}(\bot)$\\
$s'_2=P_2(s'_3)=\langle~~~~~~~~~~~~~~~~~~~~~~~~ei_{[(1,2),-]}^{(1,1)} ei _{[(2,3),-]}^{(1,2)}\rangle $ 
& $\in P_3^{-1}(s'_1)$\\
$s'_3=P_1(s'_4)=\langle~~~~~~~~~~~~~~~~~~~~~~~~ei_{[(1,2),-]}^{(1,1)} ei _{[(2,3),-]}^{(1,2)} 
                     ed_{[(2,3),-]}^{(2,1)}\rangle$ 
& $\in P_2^{-1}(s'_2)$\\
$s'_4=P_1(s'_5)=\langle vi_{[1,A]}^{(1,1)}~~~~~~~~~~~~~~~~ei_{[(1,2),-]}^{(2,1)} 
                     ei _{[(2,3),-]}^{(2,2)} ed_{[(2,3),-]}^{(3,1)}\rangle$
& $\in P_1^{-1}(s'_3)$\\
$s'_5=P_1(s_6)=\langle vi_{[1,A]}^{(1,1)} vi_{[2,B]}^{(2,1)}~~~~~~~~ 
                     ei_{[(1,2),-]}^{(3,1)} ei _{[(2,3),-]}^{(3,2)} ed_{[(2,3),-]}^{(4,1)}\rangle$
& $\in P_1^{-1}(s'_4)$\\
$s_6~~~~~~~~~~~~=\langle vi_{[1,A]}^{(1,1)} vi_{[2,B]}^{(2,1)} 
                      vi_{[3,C]}^{(3,1)} ei_{[(1,2),-]}^{(4,1)} 
                      ei _{[(2,3),-]}^{(4,2)} ed_{[(2,3),-]}^{(5,1)}\rangle$ 
& $\in P_1^{-1}(s'_5)$ \\
&\\
\hline
\end{tabular}
\end{center}
\caption{Parent-child relation between rFTSs.}
\label{GTRACEprocedure2}
\end{figure*}

\begin{example}
Since the rFTS $s_6$ shown in Fig.~\ref{GTRACEprocedure2} contains TRs that are applied to vertices, $P_1$ is applied to $s_6$ three times until the transformation sequence does not contain TRs applied to vertices, and we obtain $s'_5$, $s'_4$, and $s'_3$. Next, $P_2$ is applied to $s'_3$ once, resulting in $s'_2$ where each TR is applied to a mutually different edge in $g_u(P_2(s'_2))$. Finally, $P_3$ is applied to $s'_2$ twice until the transformation sequence becomes $\bot$. 
\end{example}

We have defined our parent-child relation between rFTSs in terms of three functions $P_1$, $P_2$, and $P_3$. Using the search tree $T$ over $S$ defined by $P=\{P_1, P_2, P_3\}$, we can enumerate a complete set of only rFTSs of $S$ by traversing search tree $T$ from its root. In the traversal, $P_1^{-1}$, $P_2^{-1}$, and $P_3^{-1}$ are used to enumerate rFTSs from the current rFTS by adding~a~TR.


\subsection{GTRACE-RS}
In the previous subsection, we defined functions $P_1$, $P_2$, and $P_3$ for the parent-child relation between rFTSs. By traversing the tree constructed by parent-child relations in a reverse direction, we enumerate a complete set of rFTSs without enumerating any FTSs that are not also rFTSs. Figure~\ref{GTRACE3} gives the pseudo code for the proposed method ``GTRACE-RS'' that mines all rFTSs from graph sequences $DB$ and accumulates them in $S$ with the minimum support value $\sigma'$. The method traverses the search tree $T$ over $S$ defined by $P=\{P_1, P_2, P_3\}$ in a depth-first manner. In Fig.~\ref{GTRACE3}, $s_p \Diamond r$ means that a TR $r$ is added to $s_p$ such that $s_p \Diamond r \in P^{-1}_i(s_p)$, where $i=1,2,3$. 
In this process, a TR to be added is not always appended to the tail of $s_p$ unlike the original GTRACE.
$s_p\ne$ min checks whether $s_p$ has been discovered before, where min is the canonical form of $s_p$. 
First, given an rFTS $s_p$ that contains only TRs applied to edges and 
where each TR is applied to a mutually different edge in $g_u(s_p)$,
GTRACE-RS enumerates rFTSs that are $P_3^{-1}(s_p)$. Next, given an rFTS that contains only TRs applied to edges, GTRACE-RS enumerates rFTSs that are $P_2^{-1}(s_p)$. Finally, given an rFTS, GTRACE-RS enumerates rFTSs that are $P_1^{-1}(s_p)$.

\begin{figure}[t]
\begin{center}
\begin{tabular}{rl}
\multicolumn{2}{l}{Input: an rFTS $s_p$, a dataset $DB$, $\sigma'$, and $i=3$.}\\
\multicolumn{2}{l}{Output: the set of rFTSs $S$.}\\ \hline
\multicolumn{2}{l}{GTRACE-RS($s_p$, $DB$,  $\sigma'$, $S$, $i$)}\\ \hline
 1:  & {\bf if} $s_p\ne$ min, then \\
 2:  & ~~~~{\bf return}; \\
 3:  & insert $s_p$ into $S$;\\
 4:  & while $i > 0$ \\ 
 5:	 & ~~~~Call Subprocedure($s_p$, $DB$,  $\sigma'$, $S$, $i$); \\
 6:  & ~~~~$i=i-1$; \\
 7:  & {\bf return;} \\ \hline
	 & \\ \hline
\multicolumn{2}{l}{Subprocedure($s_p$, $DB$,  $\sigma'$, $S$, $i$)}\\ \hline
 1:  & set $C$ to $\emptyset$;\\
 2:  & scan $DB$,\\
 3:  & find all transformation rules $r$ s.t. $s_p \Diamond r \in P^{-1}_i(s_p)$;\\
 4:  & insert $s_p ~\Diamond~ r$ in $C$ and count its frequency; \\
 5:  & {\bf for each} frequent $s_p ~\Diamond~ r$ in $C$ {\bf do} \\
 6:  & ~~~~Call GTRACE-RS($s_p~\Diamond~ r$, $DB$,  $\sigma'$, $S$, $i$); \\
 7:  & {\bf return;} \\ \hline

\end{tabular}
\caption{Algorithm for enumerating rFTSs.}
\label{GTRACE3}
\end{center}
\end{figure}

\section{Implementation}
\label{Sec:Implementation}

In the previous section, we proposed a new method for mining only rFTSs from graph sequences. To mine the rFTSs efficiently, an efficient implementation of the method is also important. In this section, we explain how to implement the proposed method.  

\subsection{Implementation of $P_3^{-1}$}

Given an rFTS $s_p$ that contains only TRs applied to edges and 
where each TR is applied to a mutually different edge in $g_u(s_p)$,
GTRACE-RS enumerates rFTSs in $P_3^{-1}(s_p)$ from $s_p$. During this process, the number of edges in $g_u(P_3^{-1}(s_p))$ 
increases one by one, similarly to gSpan~\cite{yan02}. So we have implemented $P_3^{-1}$ based on gSpan\footnote{Despite having defined canonical codes of rFTSs and implemented $P^{-1}_3$ using gSpan, we could have implemented these similarly using FSG~\cite{kuramochi2001} or Gaston~\cite{nijssen04}, which are algorithms for solving the frequent graph mining problem.}. Since gSpan is an algorithm for mining frequent subgraph patterns from colored graphs\footnote{We use the term ``color'' instead of ``label'' to distinguish labels of graphs in frequent graph mining problems and labels of graph sequences in our problem.}, it cannot be applied directly to graph sequences. We extend the tuple $(u,u',c)$ of the DFS code used in gSpan to a tuple $(u,u',(l,tr,j))$ as a TR $tr_{[(u,u'),l]}^{(j,k)}$. 
An rFTS $s_p$ is extended based on the pattern-growth principle. Given a code $\alpha=(a_1,\cdots,a_k)$ for $s_p$, $\alpha$ is extended by adding a tuple $a_{k+1}$ in the form of $(u,u',(l,tr,j))$ to generate $\alpha'=(a_1,\cdots,a_k,a_{k+1})$. 
A TR, represented by the tuple, to transform an edge between vertices $u$ and $u'$ on the rightmost path in $g_u(s_p)$ can be added (backward extension), or 
a TR, represented by the tuple, to transform an edge between the rightmost vertex $u$ and another vertex $u'$ that does not exist in $g_u(s_p)$ can be added (forward extension).


\subsection{Projection}

After mining rFTSs using $P_3^{-1}$, we shrink the transformation sequences in $DB$ for the sake of computational efficiency. We call this procedure the projection. When $P_1^{-1}$ and $P_2^{-1}$ grow an rFTS $s_p$ by adding TRs one by one, the union graphs of any rFTSs generated by $P_2^{-1}$ and $P_1^{-1}$ are isomorphic with $g_u(s_p)$ according to Lemmas~\ref{lemmaP1-1} and \ref{lemmaP2-1}; that is, vertex IDs in any TRs added to $s_p$ by $P_1^{-1}$ and $P_2^{-1}$ are identical with vertex IDs in $s_p$~(A). 
In addition, when $P_2^{-1}$ grows an rFTS by adding TRs $tr_{[o,l]}^{(j,k)}$ one by one, there must exist a TR $tr_{[o',l']}^{(j',k')}$ in $s_p$ such that $o=o'$ and $j'<j$ according to Definition~\ref{def:P2} (B). 
Based on the above discussion, we define the projection as follows.

\begin{definition}
Given a transformation sequence $s_d$ of $\langle gid, d\rangle   \in DB$ and an rFTS $s_p$ mined by $P_3^{-1}$ such that $s_p \sqsubseteq s_d$ via $(\phi,\psi)$, 
we define the projection of $s_d$ onto its maximum subsequences $s'_d$ satisfying the following conditions.
\begin{itemize}
\item[(A)] $\{\psi(u),\psi(u') \mid tr_{[(u,u'),l]}^{(j,k)} \in s_p \}  = \{u,u',u'' \mid tr_{[(u,u'),l]}^{(j,k)} \in s'_d, tr_{[u'',l]}^{(j,k)} \in s'_d \}$.
\item[(B)] TRs $tr^{(j',k')}_{[o',l']} \in s_d$ such that $\psi(o)=o'$ and $\phi(j)<j'$, are included in $s'_d$ if $tr^{(j,k)}_{[o,l]} \in s_p$.~~~~~~~~~~~~~~~~~~~~~~~~~~~~~~~~~~~~~~~~~~~~~~~~~~~~~~~~~~~~~~~~~~~~~~~$\blacksquare$
\end{itemize}
\end{definition}
Projected transformation sequences are used in the implementation of $P_1^{-1}$ and $P_1^{-2}$, as explained in Section~\ref{ImpPrefixSpan}.

\subsection{Implementation of $P_1^{-1}$ and $P_1^{-2}$}
\label{ImpPrefixSpan}

\begin{figure}[t]
\centering
\includegraphics[height=58mm]{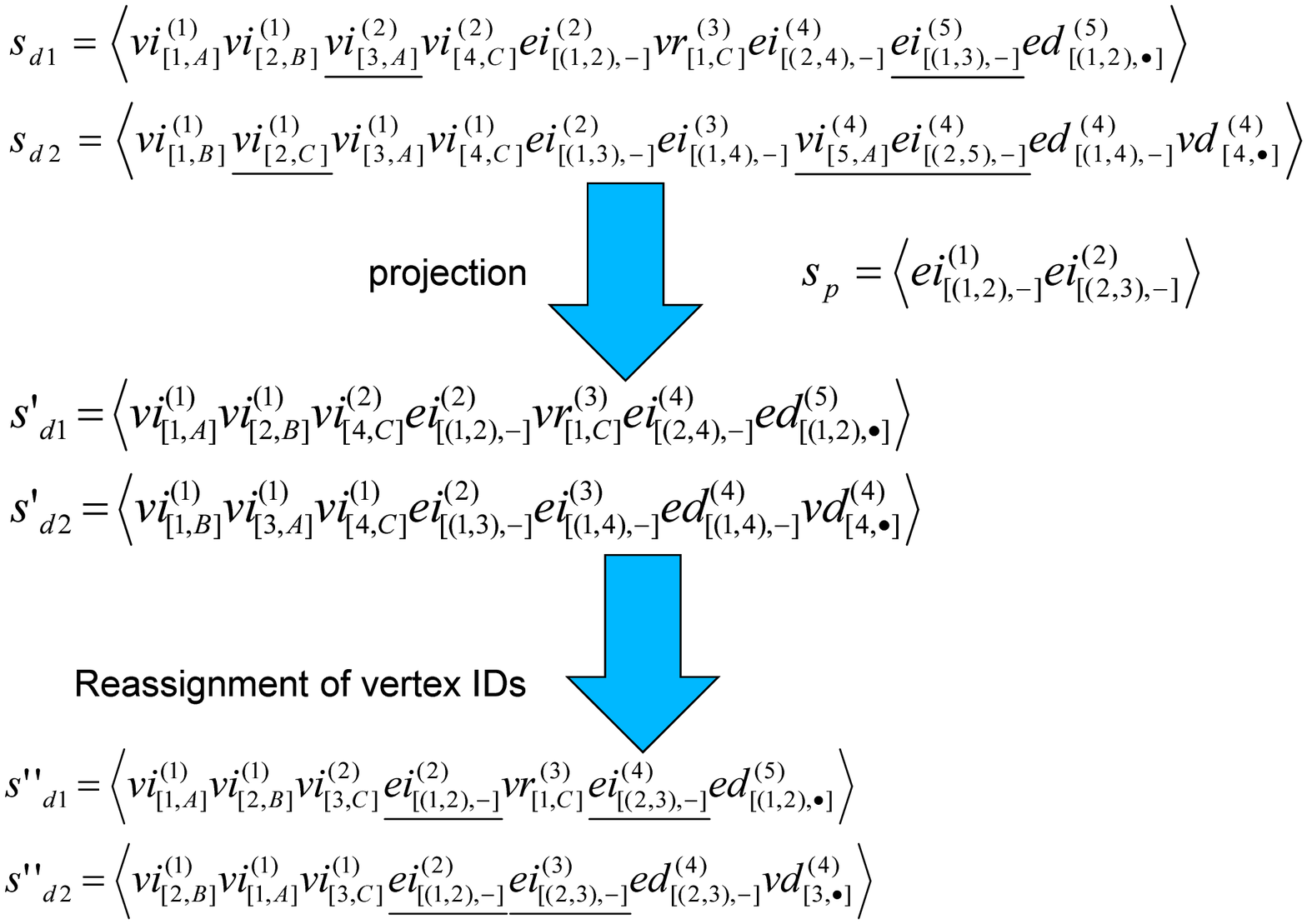}
\caption{Reassignment of vertex IDs (values for $k$ are omitted in this figure).}
\label{reassignment}
\end{figure}

In the rest of this section, using a concrete example, we explain our implementation of $P_1^{-1}$ and $P_2^{-1}$ to enumerate rFTSs. 
In Section 3, we defined $P_1$ and $P_2$ separately for ease of explanation. However, in our implementation, $P_1^{-1}$ and $P_2^{-1}$ are implemented jointly, because of the efficiency of our GTRACE-RS. For the sake of simplicity, in this subsection, we assume that $DB$ contains only two graph sequences $\langle gid_1,d_1\rangle$ and $\langle gid_2,d_2\rangle$ as shown in Fig.~\ref{reassignment}, and that an rFTS $s_p$ has been mined from $DB$ using $P_3^{-1}$ under $\sigma'=2$. In addition, we consider only two of the mappings from vertex IDs in $s_p$ to vertex IDs in the transformation sequences $s_{d1}$ and $s_{d2}$, although there are other mappings. The two mappings $\psi_1$ and $\psi_2$ are given as 
\begin{eqnarray}
\psi_1(1)=1,\psi_1(2)=2,\psi_1(3)=4, \label{map1}\\
\psi_2(1)=3,\psi_2(2)=1,\psi_2(3)=4. \label{map2}
\end{eqnarray}
In this example, vertex ID 1 in $s_p$ corresponds to vertex ID 3 in $s_{d2}$, because $\psi_2(1)=3$. Since vertex ID 3 in $s_{d1}$ and vertex IDs 2 and 3 in $s_{d2}$ do not correspond to any vertex IDs in $s_p$, the underlined TRs in $s_{d1}$ and $s_{d2}$ in Fig.~\ref{reassignment} are removed in the projection, and the projected transformation sequences $s'_{d1}$ and $s'_{d2}$ shown in the middle of Fig.~\ref{reassignment} are derived. Transformation sequences $s'_{d1}$ and $s'_{d2}$ are used in $P_1^{-1}$ and $P_2^{-1}$ to grow $s_p$. Therefore, all rFTSs mined from $s'_{d1}$ and $s'_{d2}$ must contain $s_p$ as a subsequence.

Subsequently, to check efficiently whether TRs in $s'_{d1}$ and $s'_{d2}$ correspond to each other, we convert the projected transformation sequence $s'_{d1}$ and $s'_{d2}$ by reassigning all vertex IDs in the projected transformation sequences, as described below. We are aware of all the embeddings of $s_p$ in $s'_{d1}$ and $s'_{d2}$, as well as their mappings from vertex IDs in $s_p$ to vertex IDs in $s'_{d1}$ and $s'_{d2}$ as given by Eqs.~(\ref{map1}) and (\ref{map2}). In this example, we reassign vertex ID 4 in $s'_{d1}$ to vertex ID 3, because we know $\psi_1(3)=4$. The reassigned transformation sequences $s''_{d1}$ and $s''_{d1}$ are shown at the bottom of Fig.~\ref{reassignment}.
By reassigning the vertex IDs in $s'_{d1}$ and $s'_{d2}$, corresponding TRs are written in the same representation, except for $j$ in $tr_{[o,l]}^{(j)}$. 
For example, a TR $vi_{[2,B]}$ is written in the same representation in $s''_{d1}$ and $s''_{d2}$, although the rule is written as $vi_{[2,B]}$ and $vi_{[1,B]}$ in $s'_{d1}$ and $s'_{d2}$, respectively. 

Next, because the corresponding rules are written in the same representation, we further convert the reassigned transformation sequences $s''_{d1}$ and $s''_{d2}$ in Fig.~\ref{reassignment} to sequences of itemsets as follows. 
\begin{eqnarray}
&&s_1=\langle gid_1,\langle (i_1~i_2)~(i_3~\underline{i_4})~i_5~\underline{i_6}~i_7 \rangle \rangle, \label{s1}  \\
&&s_2=\langle gid_2,\langle (i_1~i_2~i_3)~\underline{i_4}~\underline{i_6}~(i_8~i_9) \rangle \rangle, \label{s2}
\end{eqnarray}
where $vi_{[1,A]}$ in $s''_{d1}$ and $s''_{d2}$ is converted to $i_1$, $vi_{[1,B]}$ to $i_2$, and so on. In addition, items in $s_1$ and $s_2$ are placed in the same parentheses if their corresponding TRs exist in the same intrastate transformation sequence. For brevity, parentheses are omitted if an itemset has only one item. In these sequences of itemsets, an underlined item appears in all of the converted sequences as a subsequence $\langle i_4~i_6\rangle$, because $s=\langle i_4~i_6\rangle$ corresponds to $s_p$. To shorten the converted sequences, we remove the subsequence $s$ from the converted sequences to obtain other converted sequences as follows.
\begin{eqnarray}
&&s'_1=\langle gid_1,\langle (i_1^{<1st}~i_2^{<1st})~i_3^{=1st}~i_5^{<2nd}~i_7 \rangle \rangle,  \nonumber \\
&&s'_2=\langle gid_2,\langle (i_1^{<1st}~i_2^{<1st}~i_3^{<1st})~(i_8~i_9) \rangle \rangle, \nonumber 
\end{eqnarray}
where $i_3^{<1st}$ denotes that $i_3$ appears before the 1st item in $s$, and $i_3^{=1st}$ denotes that $i_3$ appears at the same time as the 1st item in $s$. 
Using this representation and $s$, we can re-convert $s'_1$ and $s'_2$ to $s_1$ and $s_2$, respectively. Now, we have two sequences of itemsets $s'_1$ and $s'_2$. By applying PrefixSpan~\cite{pei01} to the sequences of itemsets under $\sigma'=2$, we obtain three frequent sequential patterns $\{\langle i_1^{<1st} \rangle, \langle i_2^{<1st}\rangle, \langle(i_1^{<1st}~i_2^{<1st}) \rangle\}$\footnote{Although 17 frequent sequential patterns are mined from $s_1$ and $s_2$, some of these are not rFTSs. We discuss this in detail in Section~6.}. After mining the frequent sequential patterns, we re-convert each item in the frequent sequential patterns to a TR to obtain three rFTSs: 
$\langle vi^{(1)}_{[1,A]} ei^{(2)}_{[(1,2),-]} ei^{(3)}_{[(2,3),-]}\rangle$, 
$\langle vi^{(1)}_{[2,B]} ei^{(2)}_{[(1,2),-]} ei^{(3)}_{[(2,3),-]}\rangle$, and 
$\langle vi^{(1)}_{[1,A]} vi^{(1)}_{[2,B]} ei^{(2)}_{[(1,2),-]} ei^{(3)}_{[(2,3),-]}\rangle$.

In PrefixSpan, we can check in $O(1)$ whether TRs in $s''_{d1}$ and $s''_{d2}$ correspond to each other, because two corresponding TRs in the two reassigned transformation sequences have the same vertex IDs. By reassigning vertex IDs in projected transformation sequences and converting transformation sequences to sequences of itemsets, we avoid graph isomorphism matching between two transformation sequences. Since we can quickly check whether two TRs correspond by comparing two items, we can efficiently mine rFTSs from large and long graph sequences.

\section{Experiments}
\label{Section:exp}

The proposed method was implemented in C++. The experiments were executed on an HP Z600 computer with an Intel Xeon X5560 2.80GHz processor, 4 GB of main memory and running Windows 7. The performance of the proposed method was evaluated using artificial and real-world graph sequence data. 


\subsection{Artificial Datasets}

\begin{table}[t]
\begin{center}
\caption{Parameters of the artificial datasets.}
\label{default}
\begin{tabular}{ll}\hline 
Parameters 									& Default values\\ \hline 
Probability of vertex and edge insertions 	& \\
~~~~~~~~~~~~~in transformation sequences	& $p_i=$		80\%\\
Probability of vertex and edge deletions  	& \\
~~~~~~~~~~~~~in transformation sequences	& $p_d=$		10\%\\
Average number of vertex IDs 				& \\
~~~~~~~~~~~~~in transformation sequences	& $|V_{avg}|=$ 6\\
Average number of vertex IDs 				& \\
~~~~~~~~~~~~~in embedded FTSs				&$|V'_{avg}|=$	3  \\ 
Number of vertex labels						& $|L_v|=$		5 \\ 
Number of edge labels						& $|L_e|=$		5 \\ 
Number of embedded FTSs						& $N=$			10 \\ 
Number of transformation sequences			& $|DB|=$		1,000 \\ 
Edge existence probability between vertices & $p_e=$15\%\\
Average edit distance between interstates   & $d_{ist}=$2\\
Minimum support threshold					& $\sigma'=$	10\% \\ 
\hline
\end{tabular}
\end{center}
\end{table}

We compared the performance of the proposed method with the original GTRACE~\cite{inokuchi2008} using artificial datasets generated from the parameters listed in Table~\ref{default}. 
First, starting from $g^{(1)}$ with $|V_{avg}|/2$ vertices generated with edge existence probability $p_e$, we grew each graph sequence to include $|V_{avg}|$ vertex IDs on average, by applying two ($=d_{dist}$) of (A) inserting with probability $p_i$, (B) deleting with probability $p_d$, and (C) relabeling vertices and edges with probability $1-p_i-p_d$ at each interstate. Accordingly, if $p_i$ is small or $|V_{avg}|$ is large, the generated transformation sequence is long. This process is continued until the sequence becomes relevant by increasing the numbers of vertices and edges. We generated $|DB|$ graph sequences. Similarly, we generated $N$ rFTSs with $|V'_{avg}|$ vertex IDs on average. We then generated the $DB$ in which each graph sequence was overlaid by an rFTS with probability $1/N$. Each graph sequence contained $|L_v|$ vertex labels and $|L_e|$ edge labels.

Table~\ref{expLabel} lists the computation times [sec], the numbers of rFTSs mined by the proposed method, and the numbers of FTSs mined by the first step of the original GTRACE
for varying values of $|DB|$, $|V_{avg}|$, $p_i$, $L_e$, and $\sigma'$, with the other parameters set to their default values. In the table, ``--'' indicates that no results were obtained because of intractable computation times exceeding 2 hours. In addition, PM and GT denote the proposed method and GTRACE, respectively.

\begin{table}[b]
\begin{center}
\caption{Results for various $|DB|$, $|V_{avg}|$, $p_i$, $|L_e|$, and $\sigma'$.}
\begin{tabular}{l|rrrr}\hline 

$|DB|$			&1,000	&3,000	&7,000	&10,000	\\ \hline
avg. len. 		&42.9	&44.0	&43.5	&43.4	\\ \hline
PM comptime		&2.0	&6.5	&15.0	&21.4	\\
~\# of rFTSs	&6307	&6872	&6325	&6266	\\
GT comptime		&217.3	&951.2	&2226.2	&3465.3	\\			
~\# of FTSs		&171787	&190876	&170387	&170902	\\	

\hline \hline 
		
$|V_{avg}|$		&6		&8		&15			&20\\ \hline
avg. len. 		&42.9	&61.2	&135.5		&188.7\\ \hline
PM comptime		&2.0	&6.3	&228.7		&4335.0\\
~\# of rFTSs	&6307	&20755	&3653370	&81012875\\
GT comptime		&217.3	&3599.3	&-			&-\\
~\# of FTSs		&171787	&936115	&-			&-\\

		\hline\hline	
				
$p_i$	[\%]	&55			&70			&80		&100\\ \hline
avg. len. 		&116.7		&65.0		&42.9	&18.7\\ \hline
PM comptime		&419.0		&7.5		&2.0	&0.7\\
~\# of rFTSs	&4458046	&58251		&6307	&585\\
GT comptime		&-			&4132.2		&217.3	&8.5\\
~\# of FTSs		&-			&2355657	&171787	&9765\\

		\hline\hline	

$|L_e|$			&1			&3		&7		&10\\ \hline
avg. len. 		&43.5		&43.7	&43.8	&43.0\\ \hline
PM comptime		&54.5		&5.1	&1.9	&0.97\\
~\# of rFTSs	&70257		&12972	&5318	&3333\\
GT comptime		&2195.9		&644.0	&205.9	&122.1\\			
~\# of FTSs		&2995499	&474564	&132538	&62897\\	
				
		\hline\hline	
		
$\sigma'$ [\%]	&5			&7.5		&10		&15		\\\hline
avg. len. 		&42.9		&42.9		&42.9	&42.9	\\ \hline
PM comptime		&377.1		&30.7		&2.0	&1.8	\\
~\# of rFTSs	&90607156	&5744037	&6307	&1630	\\
GT comptime		&2122.4		&472.3		&217.3	&93.3	\\
~\# of FTSs		&176177313	&11891069	&171787	&43100	\\
\hline

\multicolumn{5}{l}{PM: Proposed Method, GT: the original GTRACE, comptime: }\\
\multicolumn{5}{l}{computation time [sec], avg. len.: average length of transformation }\\
\multicolumn{5}{l}{sequences, \# of rFTSs (or FTSs): the number of mined rFTSs}\\
\multicolumn{5}{l}{(or FTSs)}\\
\end{tabular}
\label{expLabel}
\end{center}
\end{table}

The first part of Table~\ref{expLabel} shows that the computation time is proportional to the number of graph sequences $|DB|$, as is the case in conventional frequent pattern mining.
The second and third parts of the table indicate that the computation times for both GTRACE-RS and GTRACE are exponential with respect to an increase in the average number of vertices $|V_{avg}|$ in the graph sequences and a decrease in the probability $p_i$ of vertex and edge insertions in the graph sequence. 
The main reason that the computation time increases with the average length seems to be the increase in the numbers of rFTSs in both cases.
However, the far superior efficiency of the proposed method compared to GTRACE is confirmed by the computation times. 
The fourth part of Table~\ref{expLabel} shows the effect of the number of labels on the efficiency. When $|L_e|$ is small, many transformation subsequence are isomorphic with each other, and thus the computation times for GTRACE and GTRACE-RS increase. However, the computation time for GTRACE-RS remains smaller, since it mines only rFTSs.
The fifth part of Table~\ref{expLabel} shows that the proposed method is tractable even with a low minimum support threshold.


All parts of Table~\ref{expLabel} show that the number of rFTSs mined by the proposed method is much smaller than the number of FTSs mined by GTRACE. By mining only rFTSs, the proposed method efficiently mines a complete set of rFTSs from a set of graph sequences. Therefore, the proposed method is applicable in practice to graph sequences that are both long and large.

\subsection{Real-World Dataset}

To assess the practicality of the proposed method, it was applied to the Enron Email Dataset~\cite{enron,Klimt}.
In the dataset, we assigned a vertex ID to each person participating in an email communication, and assigned an edge to a pair communicating via email on a particular day, thereby obtaining a daily graph $g^{(j)}$. In addition, one of the labels \{CEO, Employee, Director, Manager, Lawyer, President, 
Trader, Vice President\} was assigned to each vertex and we labeled each edge according to the volume of mail. 
We then obtained a set of weekly graph sequence data, {\it i.e.}, a $DB$.
The total number of weeks, {\it i.e.}, number of sequences, was 123. We randomly sampled $|V|$ ($=1 \sim 182$) persons to form each~$DB$. 

\begin{table}[t]
\begin{center}
\caption{Results for the Enron dataset.}
\begin{tabular}{l|cccc}\hline 

\# of persons $|V|$	&100	&140	&150	&182 \\ \hline
PM comptime			&0.5	&2.2	&15.9	&278.6\\
~\# of rFTSs		&33227	&31391	&47015	&1558833\\
GT comptime			&16.8	&118.3	&-		&-\\
~\# of FTSs			&66072	&154541	&-		&-\\

		\hline\hline	

min. sup. $\sigma'$[\%]	&40		&30		&20		&10 \\ \hline
PM comptime				&2.0	&6.1	&30.0	&278.6\\
~\# of rFTSs			&974	&3548	&14419	&158833\\
GT comptime				&14.4	&387.3	&-		&-\\
~\# of FTSs				&4129	&29253	&-		&-\\

		\hline\hline

\# of interstates $n$	&4		&5		&6		&7	\\ \hline
PM comptime				&5.85	&34.1	&95.6	&278.6\\
~\# of rFTSs			&5542	&21214	&51727	&158833\\
GT comptime 			&423.8	&-		&-		&-		\\ 
~\# of FTSs				&67997	&-		&-		&-\\ \hline

\multicolumn{5}{l}{Default: minimum support $\sigma'=$10\%, \# of vertex labels $|L_v|=8$,}\\
\multicolumn{5}{l}{~\# of edge labels $|L_e|=5$, \# of persons $|V|=182$, \# of interstates }\\
\multicolumn{5}{l}{~$n$=7. PM: Proposed Method, GT: the original GTRACE}\\

\end{tabular}
\label{TableEnron}
\end{center}
\end{table}

Table~\ref{TableEnron} shows the computation times (comptime [sec]) and the numbers of mined rFTSs or FTSs (\# of rFTSs or \# of FTSs) obtained for various numbers of vertex IDs (persons) $|V|$, minimum support $\sigma'$, and numbers of interstates $n$ in each graph sequence of the dataset. All the other parameters were set to the default values indicated at the bottom of the table. 
Thus, the dataset with the default values contained 123 graph sequences each consisting of 182 persons (vertex IDs) and 7 interstates. The parameter $|l_{avg}|=$4, 5, 6, or 7 indicates that each sequence $d$ in $DB$ consists of 4, 5, 6, or 7 interstates from Monday to Thursday, Friday, Saturday, or Sunday, respectively. 

The upper, middle, and lower parts of the table show the practical scalability of the proposed method with regard to the number of persons (vertex IDs), the minimum support threshold, and the number of interstates in graph sequences in the graph sequence database, respectively. The original GTRACE proved intractable for the graph sequence dataset generated from the default values, despite the change in each graph in this graph sequence database being gradual. On the other hand, execution of the proposed method is tractable with respect to the database. 
Good scalability of the proposed method is indicated in Table~\ref{TableEnron}, because the computation times for the proposed method are smaller than those for the original GTRACE. 
The scalability of the proposed method comes from mining only rFTSs based on the principle of a reverse search and the efficient implementation as discussed in Section~4. 


\section{Discussion}
\label{Discussion}


In Section~\ref{ImpPrefixSpan}, we mentioned that 17 frequent sequential patterns are mined from $s_1$ and $s_2$ given by Eqs.~(\ref{s1}) and~(\ref{s2}), respectively.
Some of the mined patterns from $s_1$ and $s_2$ are given below. 
\begin{eqnarray}
&&\langle i_1 i_6\rangle =\langle vi^{(1)}_{[1,A]}ei^{(2)}_{[(2,3),-]}\rangle, \label{notRFTS1}\\
&&\langle (i_1 i_2)\rangle =\langle vi^{(1)}_{[1,A]}vi^{(1)}_{[2,B]}\rangle, \label{notRFTS2}\\
&&\langle (i_1 i_2) i_4\rangle =\langle vi^{(1)}_{[1,A]}vi^{(1)}_{[2,B]}ei^{(2)}_{[(1,2),-]}\rangle, \\
&&\langle (i_1 i_2) i_6\rangle =\langle vi^{(1)}_{[1,A]}vi^{(1)}_{[2,B]}ei^{(2)}_{[(2,3),-]}\rangle. \label{notRFTS3}
\end{eqnarray}
Fourteen frequent sequential patterns, including Eqs.~(\ref{notRFTS1}) to (\ref{notRFTS3}), of the 17 patterns should not be mined from the projected transformation sequences with respect to $s_p=\langle ei_{[(1,2),-]}^{(1)}  ei_{[(2,3),-]}^{(2)}\rangle$, because the transformation sequences shown as Eqs.~(\ref{notRFTS1}) to (\ref{notRFTS3}) do not contain $s_p$ as a proper subsequence according to the principle of the reverse search. In addition, the transformation sequences shown as Eqs.~(\ref{notRFTS1}), (\ref{notRFTS2}), and (\ref{notRFTS3}) are not rFTSs. By converting $s_1$ and $s_2$ to $s'_1$ and $s'_2$, respectively, we mine only rFTSs that should be mined from the projected transformation subsequences.


\section{Conclusion}
\label{Section:Conclusion}

In this paper, we proposed an efficient method for mining all rFTSs from a given set of graph sequences. We developed a graph sequence mining program, and confirmed the efficient and practical performance of the proposed method through computational experiments using artificial and real-world datasets.
The method proposed in this paper efficiently enumerates all rFTSs from a set of graph sequences, whereas the methods in~\cite{Borgwardt, Berlingerio} mine all frequent patterns from a long graph sequence. In \cite{kuramochi,fiedler}, it is shown that the principle of growing possible patterns can be distinguished from the principle of counting support values of the patterns. Therefore, the proposed method in this paper can be extended to mine rFTSs from a long and large graph sequence based on \cite{kuramochi,fiedler}. By extending our method to mine from graph sequences, we plan to compare the performance of our method with that of Berlingerio's recently proposed~method.


\end{document}